\documentclass{article}
\usepackage{iclr2027_conference,times}

\usepackage{amsmath,amsfonts,bm}

\def\eqref#1{equation~\ref{#1}}

\def\1{\bm{1}}

\DeclareMathAlphabet{\mathsfit}{\encodingdefault}{\sfdefault}{m}{sl}
\SetMathAlphabet{\mathsfit}{bold}{\encodingdefault}{\sfdefault}{bx}{n}

\usepackage{hyperref}
\usepackage{url}
\usepackage{booktabs}
\usepackage{amsmath}
\usepackage{amssymb}
\usepackage{amsthm}
\usepackage{graphicx}
\usepackage{xcolor}
\usepackage{xspace}
\usepackage{pifont}
\usepackage{microtype}

\newcommand{\method}{\textbf{GLIE}\xspace}
\newcommand{\cmark}{\ding{51}}
\newcommand{\xmark}{\ding{55}}

\newtheorem{proposition}{Proposition}

\title{Generative Late-Interaction Embeddings For \\ Visual Document Retrieval}

\iclrfinalcopy 
\author{%
\makebox[\textwidth][c]{%
\begin{tabular}{@{}c@{\hspace{2.2em}}c@{\hspace{2.2em}}c@{\hspace{2.2em}}c@{}}
Mohamed Eltahir$^{1*}$ & Talal Aloushan$^{1*}$ & Rose Khairoalsendi$^{1*}$ & Jana Shata$^{1}$\\[0.25em]
Mohammed Alhassan$^{1}$ & Leen Alrehaili$^{1}$ & Tanveer Hussain$^{2\ddagger}$ & Naeemullah Khan$^{1\S}$
\end{tabular}}\\[1.0em]
\makebox[\textwidth][c]{$^{1}$King Abdullah University of Science and Technology (KAUST), Thuwal, Saudi Arabia}\\
\makebox[\textwidth][c]{$^{2}$Department of Computer Science, Edge Hill University, Ormskirk, England}\\[0.3em]
\makebox[\textwidth][c]{\small\texttt{\{mohamed.hamid, tlal.aloushan, rouz.alsindi, jana.shata,}}\\
\makebox[\textwidth][c]{\small\texttt{mohammed.alhassan.1, leen.alrehaili, naeemullah.khan\}@kaust.edu.sa}}\\
\makebox[\textwidth][c]{\small\texttt{hussaint@edgehill.ac.uk}}
}

\begin{document}

\vspace*{-30pt}
\maketitle
\lhead{Preprint.}

{\renewcommand{\thefootnote}{\fnsymbol{footnote}}
\footnotetext[1]{Equal contribution.}
\footnotetext[3]{Corresponding author.}
\footnotetext[4]{Principal Investigator (PI).}
}

\vspace{-23pt}
\begin{figure}[h]
\centering
\begin{minipage}[h]{0.5\linewidth}\centering
\includegraphics[width=\linewidth]{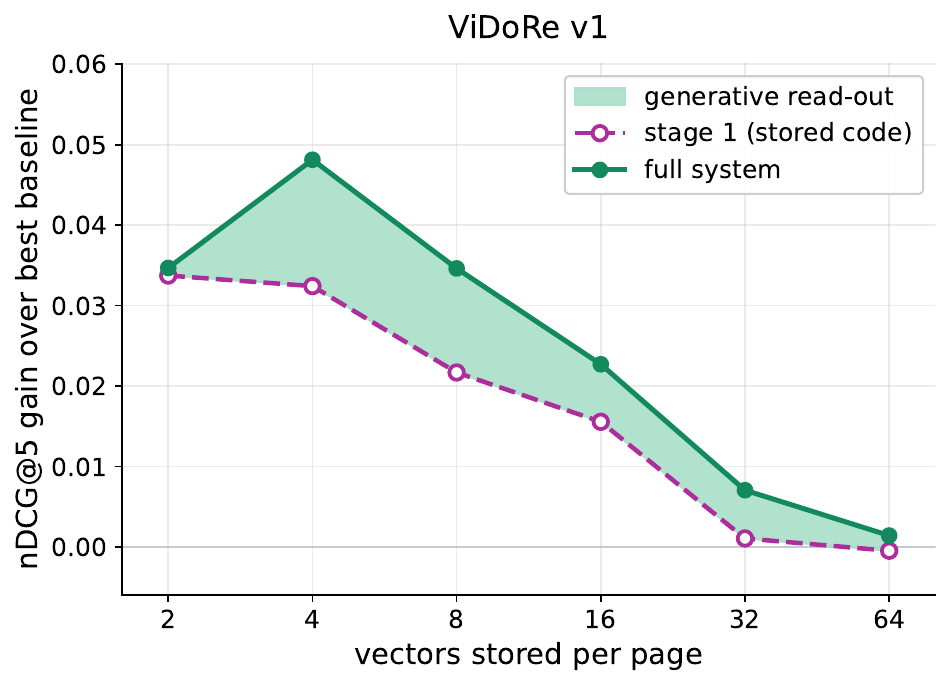}\\[-1pt]
{\small (a)}
\end{minipage}\hfill
\begin{minipage}[h]{0.5\linewidth}\centering
\includegraphics[width=\linewidth]{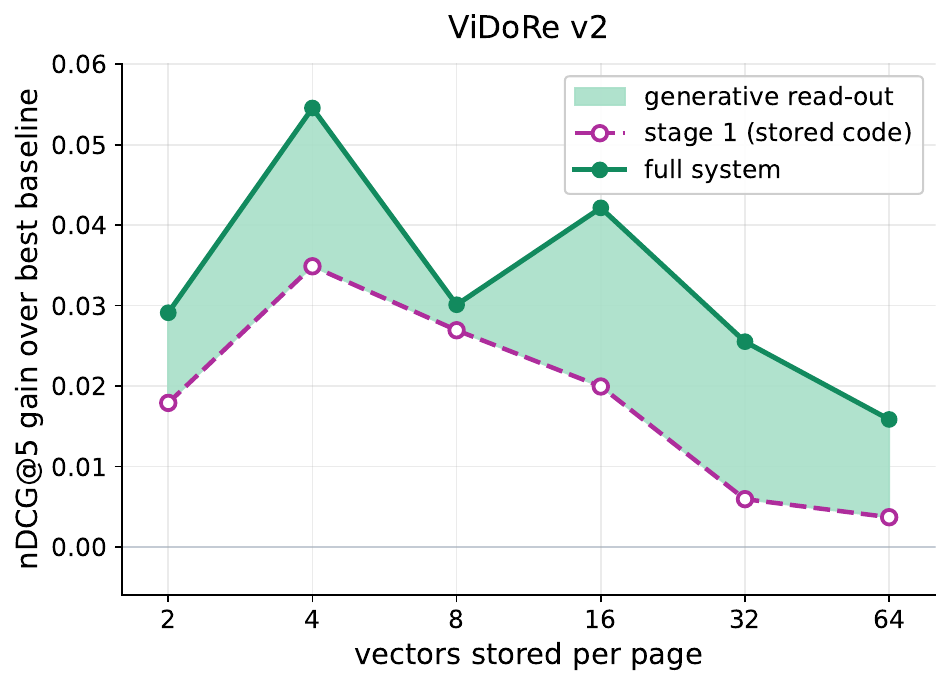}\\[-1pt]
{\small (b)}
\end{minipage}
\caption{Margin over the strongest training-free baseline at each budget on ViDoRe v1 and v2.}
\label{fig:margin}
\end{figure}

\vspace*{-0.1cm}
\begin{abstract}
Late-interaction retrieval is the state-of-the-art for visual document search, but it pays for its accuracy in storage. Existing compression methods retain a subset or local average of the $N \approx 1{,}000$ vectors per page. Under aggressive storage budgets, however, these methods degrade sharply, and alternatives require retraining the encoder. Investigating this degradation across three encoders, we found two consistent properties: the vectors lie exactly on the unit sphere and concentrate near a manifold of intrinsic dimension five to six. This geometry yields two insights. First, standard $k$-means centroids fall inside the sphere, causing systematic underestimation of MaxSim scores. Normalizing them to the surface is a free correction worth up to $+0.093$ nDCG@5 over raw centroids. Second, because the page manifold has few degrees of freedom, the full set of vectors can be regenerated from only a few. To this end, we introduce Generative Late-Interaction Embeddings (\method): $k \ll N$ vectors per page learned from the normalized centroids to serve as both a lightweight index and a basis for regenerating the page's full embedding set. At query time, search runs exclusively on these $k$ vectors, and a decoder expands only the top candidates back to all $N$ vectors for exact rescoring. At four vectors per page on ViDoRe v1, \method retains nearly 80\% of the uncompressed system's nDCG@5, against 70\% for the best prior post-hoc method. These results use a 415K-parameter network fitted in under three GPU-minutes on just a thousand training pages. At a matched training budget, fine-tuning the encoder does not reach even the training-free stage of \method, and the full system beats it at every budget. These patterns hold across a second encoder and ViDoRe v2. By reconstructing evidence on demand rather than sampling it, \method opens a new axis for storage-efficient retrieval, with the decoder as its main design surface. Code available at: \href{https://github.com/mohammad2012191/GLIE}{LINK}
\end{abstract}

\begin{table}[h]
\centering
\small
\caption{Positioning against prior compression for late interaction. The frozen-backbone column sets our experimental scope: a method that updates the encoder cannot be adopted without re-encoding the corpus, so Section~\ref{sec:results} compares against the frozen rows and reproduces only Light-ColPali, at a matched training budget.}
\label{tab:positioning}
\begin{tabular}{lccccc}
\toprule
 & Frozen & Effective & Budget-elastic & Generative &  Adapt cost \\
\textbf{Method} & backbone & $\le16$ vec & (no re-encode) & read-out & per budget \\
\midrule
Token pooling / merging & \cmark & \xmark & \cmark & \xmark &  $0$ \\
Light-ColPali (finetuned) & \xmark & \xmark & \xmark & \xmark &  $\sim$72 GPU-h \\
MetaEmbed & \xmark & \cmark & \cmark & \xmark &  $\sim$192 GPU-h \\
\midrule
\method, training-free stage & \cmark & \cmark & \cmark & \xmark & $0$ \\
\method, full & \cmark & \cmark & \cmark & \cmark  & $\sim$3 GPU-min \\
\bottomrule
\end{tabular}
\end{table}

\vspace*{-1cm}

\section{Introduction}\label{sec:intro}

Retrieval over visual documents has converged on late interaction. Instead of collapsing a page into one vector, models such as ColPali~\citep{faysse2025colpali} embed every image patch separately and score a query against a page with the MaxSim operator, the sum over query tokens of the maximum inner product with any page vector. This design is what makes visual retrieval work. It preserves local evidence such as a table cell or a phrase in a figure caption, and multi-vector representations are provably more expressive than any single vector of comparable size~\citep{jayaram2026multi}.

The price is storage. ColPali stores 1{,}031 patch vectors of dimension 128 per page, about 258 KB in bfloat16, so one million pages cost a quarter terabyte of embeddings before any index structure is added. The cost is paid at rest, in memory during search, and in transfer at query time, and the community now names index footprint among the paradigm's central open problems~\citep{clavie2026lir}. The dominant remedies compress the stored set directly: pooling merges similar vectors~\citep{clavie2024reducing}, pruning keeps a salient subset~\citep{liu2026anchor}, merging studies combine both~\citep{ma2025lightcolpali, yan2026sculpting}, and quantization stores vectors as cluster IDs and quantized residuals~\citep{santhanam2022plaid}. These methods share a premise, that the compressed representation is a subset or local average of the encoder's vectors, and they also share a limit. None reports operating below roughly sixteen vectors per page, and methods that reach smaller budgets retrain the encoder~\citep{macavaney2025efficient, veneroso2025crisp, xiao2026metaembed}, which invalidates every embedding already computed.

This paper starts from a different question: \emph{what is the stored object, geometrically?} Measuring ColPali page embeddings across all ten evaluation corpora, we find that the 1{,}031 vectors of a page concentrate near a manifold of intrinsic dimension five to six, living on the unit sphere of the 128-dimensional embedding space. Two further encoders give the same answer, one of them in a 3{,}072-dimensional space. 

If a page is really a low-dimensional surface on the sphere, then a compressed representation should describe those degrees of freedom instead of subsampling the patches. We introduce \method, which builds such a code post hoc from a frozen encoder. Per-page cluster centroids are projected back onto the unit sphere, where the encoder actually places its outputs, a free correction worth up to $+0.093$ nDCG@5 on its own. A zero-initialized network then adjusts those centroids while reading the full token set, so the code starts exactly at normalized clustering and training improves it. At query time, cheap MaxSim over the projected vectors ranks every page, and then, for the top-$L$ candidates only, a shared decoder regenerates the full set of embeddings for exact rescoring.

This design shift has three structural consequences. First, the code lies on the unit sphere by construction. Second, correctness is structural: the code cannot start worse than normalized clustering, and because the decoder emits each stored vector unchanged among its outputs, regeneration can add evidence but cannot destroy it. Third, nothing in the codec conditions on a corpus, so one fit serves every collection.

Because what \method targets is the gap between a compressed code and its own uncompressed system, a stronger encoder raises the ceiling it is closing on, and a better decoder closes more of that gap. 

\newpage
\textbf{Contributions:}

\textbf{Geometry of the stored object.} We measure page token clouds as a manifold of intrinsic dimension five to six on the unit sphere.

\textbf{Spherical anchoring.} We identify a systematic MaxSim underestimate in standard $k$-means clustering and remove it at zero cost by re-projecting centroids to the unit sphere.

\textbf{Generative Late-Interaction Embeddings (\method).} We replace extractive subset sampling with a codec that stores $k$ vectors per page and regenerates all $N$ from them on demand, together with the asymmetric pipeline it enables: initial retrieval runs on the stored vectors alone, and only a top-$L$ shortlist is expanded and rescored exactly.

\section{Related Work}\label{sec:related}

\textbf{Late-interaction retrieval.}
ColBERT introduced token-level document representations scored with MaxSim~\citep{khattab2020colbert}, ColBERTv2 compressed the index with centroid-plus-residual storage~\citep{santhanam2022colbertv2}, and ColPali transferred the paradigm to visual documents with the ViDoRe benchmark, on which late interaction dominates single-vector alternatives~\citep{faysse2025colpali}. Video-ColBERT extends it to video and names storage as its principal drawback~\citep{reddy2025video}. \citet{jayaram2026multi} proves multi-vector embeddings strictly more expressive than single vectors of comparable dimension, which cautions against collapsing pages into one vector. The gap in this line is that expressiveness is bought with storage, and no account exists of when that storage can be compressed safely.

\textbf{Post-hoc reduction of the stored set.}
Token pooling merges document vectors by hierarchical clustering, retaining about 97\% of quality at pool factor 4~\citep{clavie2024reducing}. For visual documents, Light-ColPali finds that merging post-projector embeddings retains 97.8\% of nDCG@5 at merging factor 9 and 93.6\% at factor 49, roughly a $41\times$, smaller index, its most aggressive published point being about twenty-one vectors per page~\citep{ma2025lightcolpali}. Anchor pruning~\citep{liu2026anchor} and prune-then-merge~\citep{yan2026sculpting} extend the near-lossless range. All store a subset or local average of the encoder's vectors, and none reports an operating point below roughly sixteen vectors per page, despite selecting, averaging, and pruning by different criteria. Section~\ref{sec:method-geometry} characterizes the geometry behind that shared limit, and our method targets the budgets below it.

\textbf{Small budgets via retraining.}
ConstBERT learns a projection to a constant number of document vectors~\citep{macavaney2025efficient}, CRISP trains clusterability into the representation~\citep{veneroso2025crisp}, and MetaEmbed trains nested meta-tokens serving budgets 1 to 64 at test time~\citep{xiao2026metaembed}. Each trains the encoder, so adopting one means giving up the public checkpoint and re-encoding the corpus. The training is substantial on its own: MetaEmbed reports 32 H100 GPUs for 30 hours. Budgets are also fixed at training time, and within that set a deployed index can be truncated freely, but reaching a budget outside it requires training again. Our setting is complementary: a frozen public checkpoint, cached embeddings, and a per-budget fit measured in GPU-minutes.

\textbf{Orthogonal axes.}
PLAID~\citep{santhanam2022plaid} and EMVB~\citep{nardini2024efficient} compress each vector and prune candidates at query time while leaving the per-page vector count intact, Matryoshka learning shrinks the dimension axis~\citep{kusupati2022matryoshka,xiang2026mm}, and MUVERA sketches multi-vector scoring for candidate generation while keeping full vectors at rest~\citep{dhulipala2024muvera}. These compose with vector-count compression.

Training-free reduction is cheap and elastic but stops at a shared floor. Retraining crosses the floor but gives up the frozen checkpoint. The orthogonal axes shrink bits, dimensions, or scoring cost. What none of the three asks is what the stored set actually is: a low-dimensional manifold on the sphere. Section~\ref{sec:method} measures that geometry and builds the code it implies. Table~\ref{tab:positioning} summarizes.

\section{Generative Late-Interaction Embeddings}\label{sec:method}

\subsection{What the stored object is}\label{sec:method-geometry}

Let $X = \{x_1, \dots, x_N\} \subset \mathbb{R}^D$ be the token vectors a frozen encoder produces for a page, with $N = 1{,}031$ and $D = 128$ for ColPali, and $Q = \{q_1, \dots, q_m\}$ the vectors for a query. The late-interaction score is
\begin{equation}
\mathrm{MaxSim}(Q, X) \;=\; \sum_{j=1}^{m} \max_{i} \, \langle q_j, x_i \rangle .
\label{eq:maxsim}
\end{equation}
A compressed representation at budget $k$ replaces $X$ with a set $C$ of $k$ vectors, scored by the same operator. Two measured properties of $X$ constrain what $C$ should be. Both are computed on the corpora and encoders of Section~\ref{sec:results-setup}, over $6{,}729$ pages.

\textbf{1) Low intrinsic dimension.} The TwoNN estimator~\citep{facco2017estimating} gives a median of $4.9$ against an ambient dimension of $128$ on ColPali. Per-corpus medians span $4.7$ to $5.1$, from scientific figures to government reports. A second encoder gives $5.1$, and a third, Nemotron v2~\citep{moreira2026nemotron}, gives $6.1$ at an ambient dimension of $3{,}072$ (Table~\ref{tab:geometry}). Ambient dimension varies by $24\times$ across the three and intrinsic dimension varies by one. The number is a property of the token cloud and not of the estimator. On the same ColPali pages, a Gaussian fitted to each page's own covariance, which has the same linear spectrum and no curvature, reads $32.2$, and uniform noise of the same size reads $61.4$.

\textbf{2) Unit norm.} All three encoders L2-normalize their output, so $X$ lies exactly on the sphere $\mathbb{S}^{D-1}$. 

\begin{table}[t]
\centering
\small
\caption{Per-page geometry (medians) for three encoders.}
\label{tab:geometry}
\begin{tabular}{lccc}
\toprule
\textbf{Statistic} & ColPali & ColQwen2 & Nemotron v2 \\
\midrule
Ambient dimension $D$ & $128$ & $128$ & $3{,}072$ \\
Token norm & $1.000$ & $1.000$ & $1.000$ \\
TwoNN intrinsic dimension & $4.9$ & $5.1$ & $6.1$ \\
\quad per-corpus range & $4.7$--$5.1$ & $4.9$--$5.5$ & $5.6$--$8.1$ \\
\bottomrule
\end{tabular}
\end{table}

A code that respects both properties is small and spherical. Clustering respects the first and not the second. A $k$-means centroid is a Euclidean mean of unit vectors, meaning it lies strictly inside the sphere. Because its norm is less than one, it systematically understates the inner products in~\eqref{eq:maxsim} for all query directions.

\begin{proposition}[The centroid norm records its cluster's spread]
\label{prop:norm}
Let $x_1, \dots, x_n$ be unit vectors with mean $c$. Then
\[
\frac{1}{n}\sum_{i=1}^{n} \lVert x_i - c \rVert^2 \;=\; 1 - \lVert c \rVert^2 ,
\]
and consequently the $k$-means objective on the sphere equals $\sum_j n_j\,(1 - \lVert c_j \rVert^2)$ over clusters $j$ with sizes $n_j$.
\end{proposition}
\begin{proof}
Expand $\lVert x_i - c\rVert^2 = 1 - 2\langle x_i, c\rangle + \lVert c\rVert^2$ and average, using $\frac1n\sum_i \langle x_i, c\rangle = \lVert c\rVert^2$.
\end{proof}

Three consequences follow. Firstly, projecting the centroids back onto the unit sphere costs nothing and should improve MaxSim by correcting this systematic underestimation. Secondly, the improvement should shrink as $k$ grows, because tighter clusters have means closer to the sphere. Thirdly, the projection moves the code outside the original data, so a centroid's score is no longer capped by the best true vector. A code that could only underestimate can now overestimate as well.

\subsection{What a learned code must add}

Normalized $k$-means is the training-free stage of \method, and it falls short in two ways that define the rest of the codec. First, it is blind beyond the mean. Each centroid summarizes its cluster by an average, so scoring-relevant structure inside the cluster is unrecoverable from the code no matter how the centroids are post-processed. Second, it is a sample. It stores $k$ points and can only ever score with those $k$ points, so a query token pointing at a region the sample misses loses its evidence. A code trained as a \emph{description} of the page can instead be expanded back into the fine token cloud when a candidate matters, re-materializing evidence that no $k$-vector sample can hold. This is why regeneration is useful at all.

\method discharges these requirements with three components, each of which provably starts at, and can only improve on, normalized $k$-means: (a) the code is initialized at it exactly, (b) is refined against the MaxSim operator itself while reading the full token set, and (c) is decoded under a guarantee that regeneration never lowers a stored score. The encoder is frozen throughout and everything happens post hoc on cached embeddings. Figure~\ref{fig:training} shows the fitting procedure and Figure~\ref{fig:pipeline} the query path.

\subsection{Spherical anchoring}
Per page, $k$-means centroids of the token set $X$ are computed and projected onto the unit sphere, $u_j = c_j / \lVert c_j \rVert$. By Proposition~\ref{prop:norm} the projection removes a systematic MaxSim underestimate at zero cost. Empirically it is worth $+0.031$ to $+0.093$ nDCG@5 on the full benchmark, shrinking as $k$ grows exactly as the proposition predicts (Table~\ref{tab:ladder}). Normalized clustering is therefore the training-free stage of \method and the floor every learned component must clear, as well as the practical one-line recommendation for any dot-product late-interaction system that clusters: normalize your centroids.

\subsection{Zero-initialized refinement that reads the page}
Defining $U = \{u_1, \dots, u_k\}$ for the anchors, a shared cross-attention module refines them against the full token set, with $U$ as queries and $X$ as keys and values,
\begin{equation}
C \;=\; \mathrm{normalize}\bigl(U + \pi_\theta\!\left(\mathrm{Attn}(U,\, X,\, X)\right)\bigr),
\label{eq:refine}
\end{equation}
where $\mathrm{Attn}(Q, K, V)$ is multi-head attention, $\mathrm{normalize}$ projects each row back to the unit sphere as in Section~\ref{sec:method-geometry}, and the output projection $\pi_\theta$ is initialized at zero, so at initialization $\pi_\theta(\cdot) = 0$ and $C = U$ exactly. Intuitively, the code begins exactly at normalized clustering and training can only move it where the objective improves, so the strongest training-free solution is the floor rather than a competitor. The correction reads $X$ directly, so the code can encode scoring-relevant structure that centroids average away, which no post-processing of the centroids alone could recover. The refiner and decoder together hold 415K parameters against the 3B backbone and run once per page at indexing time.

\subsection{Anchored generative read-out}
A shared decoder $g_\psi$ expands the $k$ stored vectors back into $N$ unit vectors for reranking. Its structure encodes three guarantees. First, count-proportional slots: cluster $j$ owns $n_j$ output slots, so the decoder reproduces the page's actual cluster structure rather than a fixed learned layout. Second, an exact anchor: slot 0 of each cluster emits the refined vector verbatim, so the decoded set contains the code set, and because MaxSim is a maximum,
\[
\mathrm{MaxSim}(Q, g_\psi(C)) \;\ge\; \mathrm{MaxSim}(Q, C) \quad \text{per query token.}
\]
Intuitively, regeneration can add evidence but structurally cannot destroy it, with no tuning. Third, bounded displacement. Each generated child starts at its cluster's anchor and moves along the surface of the sphere, never toward or away from the center, by a learned step of at most $\alpha = 0.75$, then renormalized. A child can therefore sit at most $\arctan(0.75) \approx 37^\circ$ from its anchor, so a cluster's children fill a patch around it rather than scattering across the sphere. Which child is which is encoded by its position in the cluster with fixed sine and cosine features, so no per-slot parameters are learned and one decoder serves clusters of any size.

Prior compression selects or averages among the encoder's vectors at indexing time, and that fixed set is all a query can ever be scored against. The generative read-out instead re-materializes the fine structure only when a candidate is worth the cost.

\begin{figure}[t]
\centering
{\includegraphics[width=\linewidth]{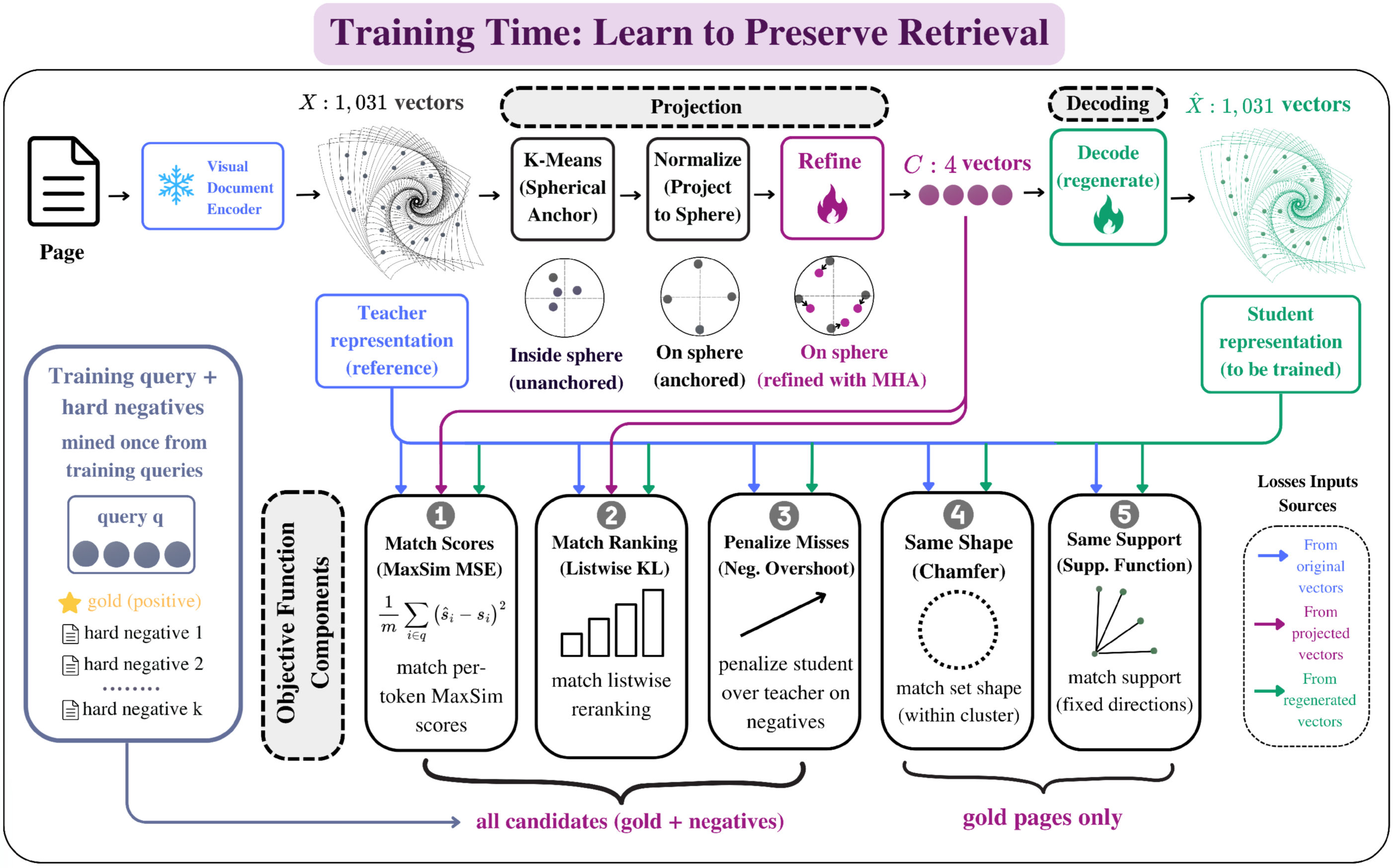}}
\caption{\textbf{Training.} The encoder stays frozen. The projection is initialized at normalized $k$-means and a refiner adjusts it from there. Five kinds of loss jointly optimize the projected and regenerated vectors, with the MaxSim and listwise terms applied to both.}
\label{fig:training}
\end{figure}

\subsection{Two-stage inference}
\method (Figure~\ref{fig:pipeline}) combines the code and the read-out:
\begin{enumerate}
\item \textbf{Stage 1 (retrieve):} score every page against the query by MaxSim over its $k$ stored vectors, at the cost of a pooled baseline.
\item \textbf{Stage 2 (regenerate and rerank):} expand the top-$L$ candidates ($L = 20$) back to $N$ vectors with $g_\psi$ and rescore them by full MaxSim. Non-candidates keep their first-stage order.
\end{enumerate}
The asymmetry is the point: cheap over everything, expensive over almost nothing.

\begin{figure}[t]
\centering
{\includegraphics[width=\linewidth]{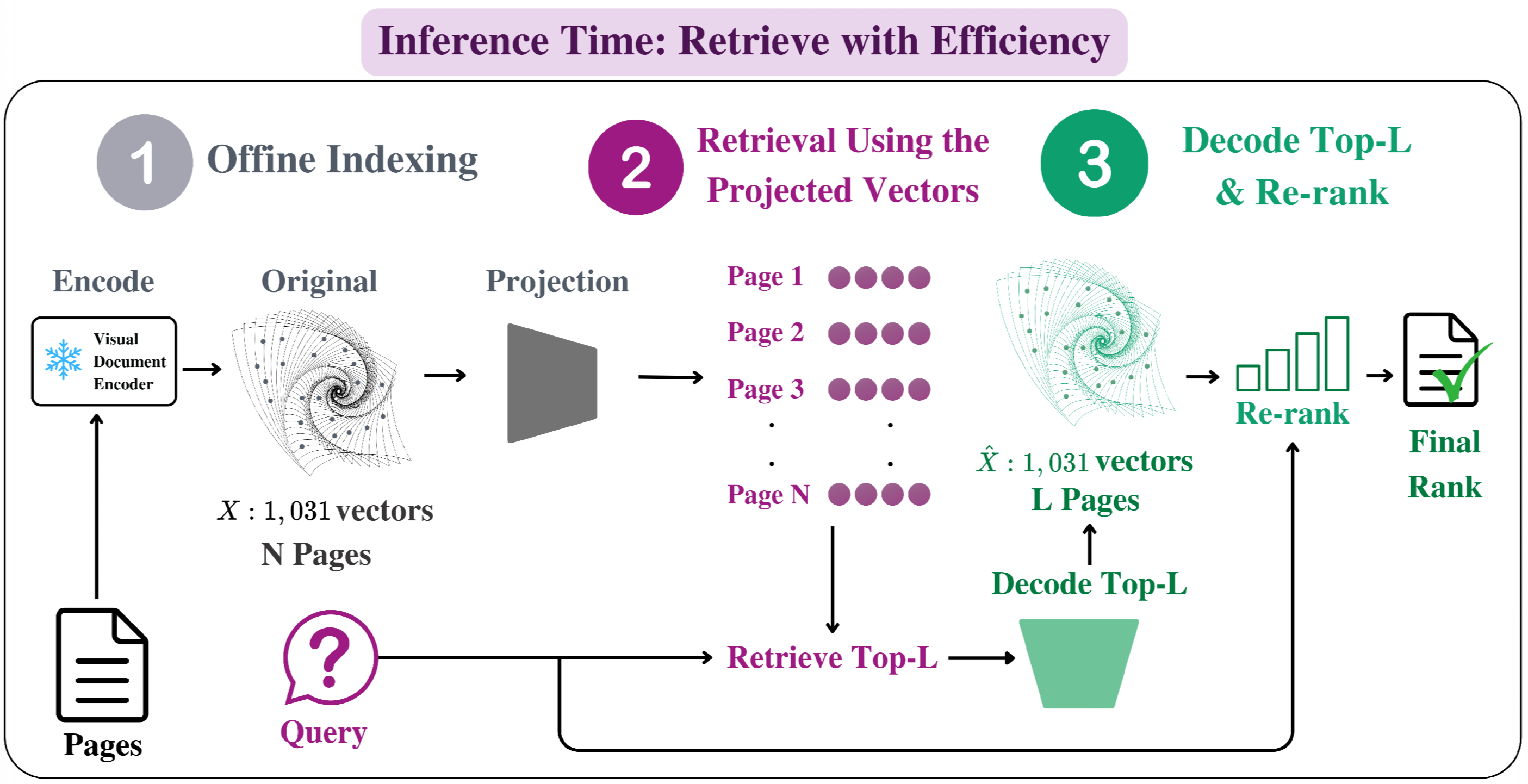}}
\caption{\textbf{Inference.} Every page is scored on its $k$ projected vectors, then only the top-$L$ shortlist is decoded back to full length and rescored exactly.}
\label{fig:pipeline}
\end{figure}

\subsection{Training}\label{sec:method-training}
Two objects are trained and each has one job. The code is scored against every page, so it must rank. The regenerated set is scored only on the shortlist, so it must rank and it must be a page. The losses follow from those two jobs.

\textbf{The code must retrieve.} Two things make a code retrieve like the full page. Firstly, every query token should find in the code the same best match it finds in the page, so we match per-query-token MaxSim values against the frozen encoder,
\begin{equation}
\mathcal{L}_{\mathrm{tok}} = \frac{1}{m}\sum_{j=1}^{m}\Bigl(\max_l \langle q_j, \hat{x}_l\rangle - \max_i \langle q_j, x_i\rangle\Bigr)^{2}.
\label{eq:distill}
\end{equation}
Here $x_i$ ranges over the page's $N$ encoder vectors and $\hat{x}_l$ over the vectors of the read-out being trained, the $k$ code vectors for the code and the $N$ regenerated vectors for the decoder, so the same term applies to both.
Secondly, the order of candidate pages should be the teacher's, so a listwise KL over each query's candidate list matches the ranking and not only the scores. Together these make the code a drop-in first-stage index.

\textbf{The regenerated set must retrieve, and it must be a page.} The same two terms apply to the regenerated vectors, since they are what the shortlist is rescored with. Three more terms enforce what a set of $N$ vectors must be to stand in for the real one. It must not invent evidence. Section~\ref{sec:method-geometry} showed that a normalized code can overestimate, and a regenerated page that scores a non-relevant candidate above its true MaxSim is exactly the error that breaks a rerank, so a one-sided penalty charges negatives only when they exceed the teacher. It must have the right shape. Each cluster's children should occupy the region its real patches occupy, in no particular order, which a Chamfer distance within each cluster of the relevant page measures. And it must reach as far as the real page in every direction. MaxSim reads the extreme point of the set along the query, so we match the support function, the farthest extent of the set along a fixed bank of random directions, between the regenerated and real vectors of the relevant page.

Reconstruction error is absent on purpose. It is minimized by placing every child near its cluster mean, which is exactly the collapse that loses the extreme points MaxSim reads. 

\section{Experiments}\label{sec:results}

\subsection{Setup}\label{sec:results-setup}

\textbf{Protocol.} We evaluate on all ten ViDoRe v1 subsets~\citep{faysse2025colpali} under the benchmark's standard protocol. The codec is fitted once on 5{,}000 pages of the public ColPali training collection, and applied frozen to every test subset. The corpus of each subset is its set of unique pages, de-duplicated by image identity, with query-less pages kept as distractors. We report single-relevant nDCG@5. 
Frozen ColPali v1.3 is used, which gives 1{,}031 patch vectors of dimension 128 per page. Budgets $k \in \{2, 4, 8, 16, 32, 64\}$, i.e.\ $516\times$ down to $16\times$ fewer stored vectors. Section~\ref{sec:results-main} repeats the sweep on ViDoRe v2~\citep{mace2025vidore} and Section~\ref{sec:results-generality} on ColQwen2. 

\textbf{Baselines.} All post hoc on the same frozen encoder at identical stored budget: raw per-page $k$-means, semantic cluster merging (the training-free stage of Light-ColPali~\citep{ma2025lightcolpali}, with its fine-tuning removed), and sequential token pooling~\citep{clavie2024reducing}. Margins are always reported against the strongest training-free baseline chosen per subset and per budget. We additionally reproduce Light-ColPali's fine-tuning stage at matched training budget and source (Section~\ref{sec:results-lcp}). Learned rows are means over three training seeds. Architecture, optimization, and evaluation details are in Appendices~\ref{app:impl} and \ref{app:protocol}.

\subsection{Main Results: ViDoRe v1 and v2}\label{sec:results-main}

\textbf{Aggressive-budget results.} \method reaches 79\% of uncompressed retrieval quality storing 1.0 KB per page and 91\% at 4 KB, and beats every prior baseline on all subsets of both ViDoRe v1 and v2 at every budget (Table~\ref{tab:main}, ~\ref{tab:v2}). The best-powered subset is also among the strongest. TAT-DQA, with 1{,}663 queries, shows margins $+0.039$ to $+0.054$ at $k \le 8$.

\textbf{Where the gains sit.} The margin is positive at every budget but not uniform, and its profile differs between the two benchmarks (Figure~\ref{fig:margin}). On v1 it forms a plateau of about $0.04$ across $k \le 8$, halves at $k=16$, and decays to noise by $k=32$, a mean of $+0.039$ below $k=8$ against $+0.010$ above. The stored code alone turns slightly negative at $k=64$, where normalized clustering already sits $0.027$ from the ceiling and a refiner fitted elsewhere has nothing left to correct. Four of the ten v1 subsets have ceilings of 0.94 to 0.98, so part of that high-budget decay is saturation rather than method failure. On v2, which saturates nowhere, the margin never decays and the band the read-out contributes stays open to $k=64$. Finally, at $k=2$, with fewer anchors than the page's intrinsic dimension, the regeneration is at its coarsest. It adds nothing on v1's results and only $+0.011$ on v2.

\subsection{What limits the cascade}\label{sec:err_decomp}

The oracle row splits the residual error (Figure~\ref{fig:headroom}). At $k=4$ on v1 the gap from \method to the oracle, 0.657 to 0.782, is decode fidelity, recoverable by better decoding, while the gap from the oracle to the ceiling, 0.782 to 0.836, is shortlist recall, recoverable only with a larger $L$. Widening the shortlist confirms the split: from $L=5$ to $L=100$ the oracle rises from 0.705 to 0.822 while \method moves only from 0.647 to 0.660, so the shortlist is not the binding constraint and $L=20$ is already past the point of diminishing returns. The generative read-out contributes $+0.013$ to $+0.016$ at $k=4$ to $8$, a third of the total margin at near-zero marginal cost, and a further $+0.13$ sits in the same shortlist waiting on a better decoder. On v2 that decode-fidelity term stays large at every budget, which is why the read-out keeps earning there. Decoding is therefore the lever with measured headroom, since refining the code further is bounded by what $k$ vectors can hold.

\begin{figure}[h]
\centering
\includegraphics[width=0.7\linewidth]{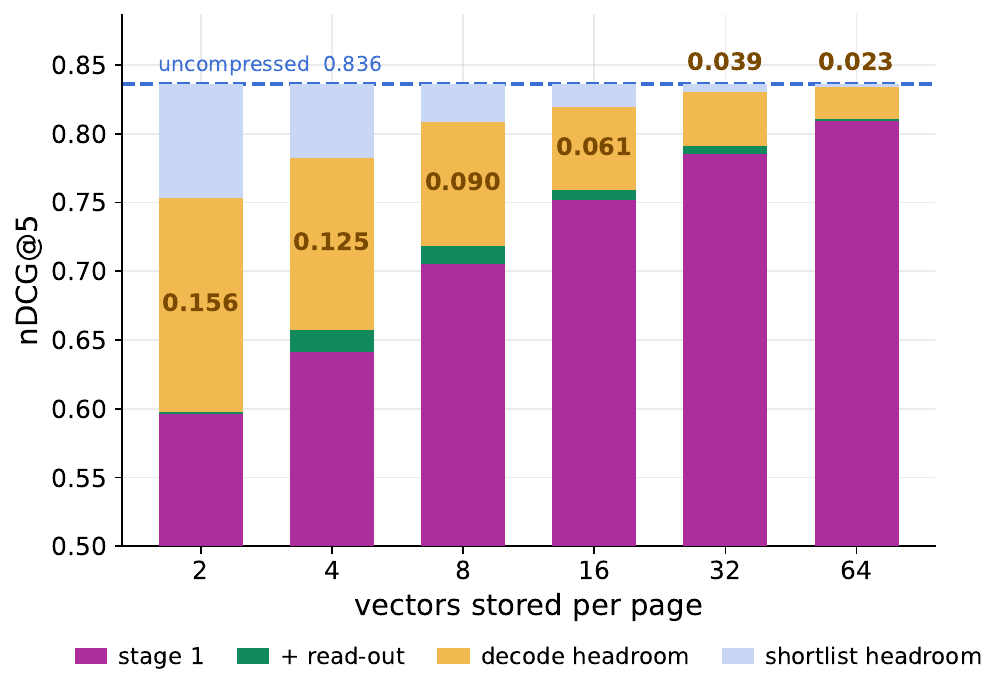}
\caption{Decomposition of the gap to the uncompressed ceiling at each budget on ViDoRe v1.}
\label{fig:headroom}
\end{figure}

\begin{table}[t]
\centering
\small
\caption{ViDoRe v1 (10 subsets), macro-averaged nDCG@5 over all queries.}
\label{tab:main}
\begin{tabular}{lcccccc}
\toprule
\textbf{Method} & $k{=}2$ & $k{=}4$ & $k{=}8$ & $k{=}16$ & $k{=}32$ & $k{=}64$ \\
\midrule
Raw $k$-means & 0.464 & 0.512 & 0.594 & 0.662 & 0.737 & 0.779 \\
Token pooling & 0.553 & 0.584 & 0.624 & 0.657 & 0.702 & 0.747 \\
Cluster merging & 0.447 & 0.471 & 0.531 & 0.613 & 0.690 & 0.763 \\
\midrule
\method & \textbf{0.597} & \textbf{0.657} & \textbf{0.718} & \textbf{0.759} & \textbf{0.791} & \textbf{0.811} \\
\quad shortlist oracle & 0.753 & 0.782 & 0.808 & 0.820 & 0.830 & 0.834 \\
\quad \% of uncompressed (0.836) & 71\% & 79\% & 86\% & 91\% & 95\% & 97\% \\
\bottomrule
\end{tabular}
\end{table}

\begin{table}[t]
\centering
\small
\caption{ViDoRe v2 (4 subsets), macro-averaged nDCG@5 over all queries.}
\label{tab:v2}
\begin{tabular}{lcccccc}
\toprule
\textbf{Method} & $k{=}2$ & $k{=}4$ & $k{=}8$ & $k{=}16$ & $k{=}32$ & $k{=}64$ \\
\midrule
Raw $k$-means & 0.184 & 0.201 & 0.275 & 0.318 & 0.376 & 0.405 \\
Token pooling & 0.228 & 0.241 & 0.244 & 0.276 & 0.307 & 0.370 \\
Cluster merging & 0.188 & 0.181 & 0.190 & 0.220 & 0.285 & 0.363 \\
\midrule
\method & \textbf{0.267} & \textbf{0.330} & \textbf{0.388} & \textbf{0.439} & \textbf{0.472} & \textbf{0.490} \\
\quad shortlist oracle & 0.394 & 0.453 & 0.485 & 0.506 & 0.509 & 0.521 \\
\quad \% of uncompressed (0.517) & 52\% & 64\% & 75\% & 85\% & 91\% & 95\% \\
\bottomrule
\end{tabular}
\end{table}

\subsection{A Small Codec Beats Fine-Tuning at a Matched Budget}\label{sec:results-lcp}

Fine-tuning the encoder on 4{,}000 pages does not reach free normalized clustering at any budget in our hands (Table~\ref{tab:lcp}), while the same source and budget spent on a frozen-backbone codec beats it at all six, by $+0.074$ to $+0.132$. The gap is not evidence against Light-ColPali, whose authors train on 130K queries for about 72 GPU-hours per budget. Instead, it measures what a small training budget buys: 13.3M LoRA parameters over 1.5 GPU-hours perturb a 3B backbone too little to help and enough to hurt, whereas 415K parameters over under three GPU-minutes reshape the stored code directly. Appendix~\ref{app:lcp} gives the adapter, merge, and training settings of the reproduction.

\begin{table}[h]
\centering
\small
\caption{Matched training source and budget, ViDoRe v1, ten subsets, all queries.}
\label{tab:lcp}
\begin{tabular}{lcccccc}
\toprule
\textbf{Method} & $k{=}2$ & $k{=}4$ & $k{=}8$ & $k{=}16$ & $k{=}32$ & $k{=}64$ \\
\midrule
Light-ColPali (LoRA fine-tune, ours) & 0.523 & 0.544 & 0.586 & 0.632 & 0.661 & 0.701 \\
Normalized $k$-means (free) & 0.552 & 0.605 & 0.684 & 0.736 & 0.784 & 0.809 \\
\method (frozen backbone) & \textbf{0.597} & \textbf{0.657} & \textbf{0.718} & \textbf{0.759} & \textbf{0.791} & \textbf{0.811} \\
\bottomrule
\end{tabular}
\end{table}

\subsection{Storage \& Data Efficiency}\label{sec:results-efficiency}

Storage at $k=4$ is 1{,}040 bytes per page against 257.8 KB uncompressed, so one million pages shrink from 258 GB to 1.0 GB. Because the backbone is frozen, changing the storage budget of a deployed index touches only cached embeddings, while a fine-tuned method must re-encode the corpus.

Refitting on 1{,}250, 2{,}500, and 5{,}000 source pages leaves the margins statistically unchanged ($+0.048/+0.052/+0.049$ at $k=4$), so it saturates on roughly a thousand pages, a hundred times smaller a slice of the collection than the encoder trains on. More data is not the lever. What the codec learns is a property of the encoder's geometry, which Table~\ref{tab:geometry} showed is corpus-independent.

\subsection{The Recipe Transfers to a Second Encoder}\label{sec:results-generality}

On ColQwen2 the codec retains 82\% of uncompressed quality at $k=4$ (Table~\ref{tab:generality}). Against spherical anchoring alone, which already reaches 0.702 at $k=4$ on this encoder, it wins at five of six budgets and loses only at $k=2$. ColQwen2 is the stronger encoder, so clustering already sits closer to its ceiling at every budget and the room a learned code has is correspondingly smaller and flatter, with margins of $+0.015$ below $k=8$ against $+0.013$ above. Read together with ViDoRe v2 (Table~\ref{tab:v2}), which saturates nowhere and keeps its margins to $k=64$, this shows that how sharply the gains concentrate at aggressive budgets is a property of the setting rather than a constant.

\begin{table}[h]
\centering
\small
\caption{Full sweep repeated on ColQwen2 over the ten ViDoRe v1 subsets.}
\label{tab:generality}
\begin{tabular}{lcccccc}
\toprule
\textbf{Method} & $k{=}2$ & $k{=}4$ & $k{=}8$ & $k{=}16$ & $k{=}32$ & $k{=}64$ \\
\midrule
Raw $k$-means & 0.524 & 0.623 & 0.706 & 0.762 & 0.806 & 0.841 \\
Token pooling & 0.619 & 0.657 & 0.702 & 0.740 & 0.774 & 0.812 \\
Cluster merging & 0.529 & 0.520 & 0.553 & 0.643 & 0.743 & 0.809 \\
\midrule
\method & \textbf{0.630} & \textbf{0.727} & \textbf{0.787} & \textbf{0.831} & \textbf{0.854} & \textbf{0.868} \\
\quad \% of uncompressed (0.883) & 71\% & 82\% & 89\% & 94\% & 97\% & 98\% \\
\midrule
Subsets improved & 7/10 & 10/10 & 10/10 & 10/10 & 10/10 & 10/10 \\
\bottomrule
\end{tabular}
\end{table}

\section{Ablation Study}\label{sec:results-ablations}

\begin{table}[t]
\centering
\small
\caption{What each component contributes, ViDoRe v1, three-seed means. Each row adds one component to the row above, so the three increments sum to the full-system gain over raw $k$-means.}
\label{tab:ladder}
\begin{tabular}{lcccccc}
\toprule
\textbf{Component} & $k{=}2$ & $k{=}4$ & $k{=}8$ & $k{=}16$ & $k{=}32$ & $k{=}64$ \\
\midrule
Raw $k$-means & 0.464 & 0.512 & 0.594 & 0.662 & 0.737 & 0.779 \\
\midrule
$+$ spherical anchoring \emph{(training-free)} & 0.552 & 0.605 & 0.684 & 0.736 & 0.784 & 0.809 \\
$+$ learned code (stage 1) & 0.596 & 0.641 & 0.705 & 0.752 & 0.785 & 0.809 \\
$+$ generative read-out (stage 2) & \textbf{0.597} & \textbf{0.657} & \textbf{0.718} & \textbf{0.759} & \textbf{0.791} & \textbf{0.811} \\
\bottomrule
\end{tabular}
\end{table}

\begin{table}[ht]
\centering

\begin{minipage}[t]{0.5\linewidth}
    \centering\small
    \caption{Decoder capacity. Fitted on 500 pages.}
    \label{tab:hyper-capacity}
\end{minipage}\hfill
\begin{minipage}[t]{0.5\linewidth}
    \centering\small
    \caption{Shortlist size at $k=4$}
    \label{tab:hyper-shortlist}
\end{minipage}

\vspace{1mm} 

\begin{minipage}[t]{0.5\linewidth}
    \centering\small
    \begin{tabular}{lccc}
        \toprule
        width $\times$ depth & params & $k=4$ & $k=16$ \\
        \midrule
        $128 \times 1$  & 184K & \textbf{0.662} & 0.755 \\
        $128 \times 3$  & 316K & 0.658 & 0.748 \\
        $256 \times 1$  & 415K & 0.660 & 0.755 \\
        $512 \times 1$  & 1.3M & 0.660 & \textbf{0.757} \\
        $512 \times 3$  & 3.4M & 0.659 & 0.756 \\
        $1024 \times 1$ & 4.6M & 0.659 & 0.751 \\
        $1024 \times 3$ & 13.0M & 0.659 & 0.753 \\
        \bottomrule
    \end{tabular}
\end{minipage}\hfill
\begin{minipage}[t]{0.5\linewidth}
    \centering\small
    \begin{tabular}{lcc}
        \toprule
        $L$ & \method & oracle \\
        \midrule
        5   & 0.647 & 0.705 \\
        10  & 0.655 & 0.750 \\
        20  & 0.658 & 0.783 \\
        50  & 0.660 & 0.810 \\
        100 & \textbf{0.660} & \textbf{0.822} \\
        \bottomrule
    \end{tabular}
\end{minipage}

\end{table}

\textbf{Component decomposition.} Table~\ref{tab:ladder} assigns the gains cleanly. Spherical anchoring is the largest single factor at every budget, from $+0.093$ at $k=4$ to $+0.030$ at $k=64$, and it is free. The learned code adds $+0.044$ to $+0.016$ up to $k=16$ and nothing beyond it. The generative read-out is the aggressive-budget specialist on v1, peaking at $+0.016$ at $k=4$. on ViDoRe v2 it does not decay, because that benchmark saturates nowhere and the headroom survives to $k=64$ (Figure~\ref{fig:margin}). 

\textbf{Hyperparameters are not where the result comes from.} Varying decoder width from 128 to 1024 and depth from one block to three, a seventy-fold range from 184K to 13M parameters, moves nDCG@5 by at most $0.009$ with no monotone trend, and the smallest decoder we train is the best of them at $k=4$ (Table~\ref{tab:hyper-capacity}). The 415K configuration we report is an arbitrary point inside a flat band rather than a tuned one. The shortlist behaves differently and is informative for a different reason (Table~\ref{tab:hyper-shortlist}). Widening $L$ from 5 to 20 is worth $+0.011$, while widening it from 20 to 100 is worth $+0.002$ even though the oracle over those same shortlists climbs from 0.705 to 0.822. The candidates are present and the decoder does not exploit them

\section{Conclusion}\label{sec:conclusion}

Late-interaction retrieval owes its accuracy to storing token-level vectors and its cost to the same decision. We argued the way out is geometric. A page's thousand token vectors concentrate near a manifold of intrinsic dimension five to six on the unit sphere, the same on every corpus and all three encoders we measure. The sphere yields a free correction any dot-product late-interaction system can adopt today. The small dimension is what makes the method possible, since a page with few degrees of freedom can be carried by a few anchors and one shared decoder, which is why 415K parameters and a thousand training pages suffice. 

We read the result as opening a new axis for the storage problem rather than closing it. Generating on demand is a different lever from storing fewer or smaller vectors, and the shortlist oracle measures how much of it is still available: at four vectors per page, decoding the same code and the same shortlist perfectly would reach $0.782$ against our $0.657$. Better decoders are therefore the immediate move on this axis, and the code itself is bounded by what $k$ vectors can hold.

Three next steps look most promising. Extending the sweep to further late-interaction encoders would test how far the recipe carries beyond the two we evaluate. Composing \method with quantized storage, since the two act on different axes of the footprint. And the largest gains should wait where token sets are largest and most redundant, which points to video late interaction~\citep{reddy2025video}.

\section{Acknowledgements}\label{sec:acknowledgements}
We are grateful to the KAUST Academy for its generous support, and especially to Prof.\ Sultan Albarakati that made this work possible. For computer time, this research used Ibex managed by the Supercomputing Core Laboratory at King Abdullah University of Science \& Technology (KAUST) in Thuwal, Saudi Arabia.

\newpage
\bibliography{bibliography}

@inproceedings{faysse2025colpali,
  title={Colpali: Efficient document retrieval with vision language models},
  author={Faysse, Manuel and Sibille, Hugues and Wu, Tony and Omrani, Bilel and Viaud, Gautier and Hudelot, C{\'e}line and Colombo, Pierre},
  booktitle={International Conference on Learning Representations},
  volume={2025},
  pages={61424--61449},
  year={2025}
}

@inproceedings{khattab2020colbert,
  title     = {{ColBERT}: Efficient and Effective Passage Search via Contextualized Late
               Interaction over {BERT}},
  author    = {Khattab, Omar and Zaharia, Matei},
  booktitle = {Proceedings of the 43rd International ACM SIGIR Conference on Research and
               Development in Information Retrieval},
  pages     = {39--48},
  year      = {2020}
}

@inproceedings{santhanam2022colbertv2,
  title     = {{ColBERTv2}: Effective and Efficient Retrieval via Lightweight Late Interaction},
  author    = {Santhanam, Keshav and Khattab, Omar and Saad-Falcon, Jon and Potts, Christopher and
               Zaharia, Matei},
  booktitle = {Proceedings of the 2022 Conference of the North American Chapter of the
               Association for Computational Linguistics: Human Language Technologies},
  pages     = {3715--3734},
  year      = {2022}
}

@inproceedings{santhanam2022plaid,
  title     = {{PLAID}: An Efficient Engine for Late Interaction Retrieval},
  author    = {Santhanam, Keshav and Khattab, Omar and Potts, Christopher and Zaharia, Matei},
  booktitle = {Proceedings of the 31st ACM International Conference on Information and
               Knowledge Management (CIKM)},
  pages     = {1747--1756},
  year      = {2022},
  note      = {arXiv:2205.09707}
}

@inproceedings{reddy2025video,
  title={Video-colbert: Contextualized late interaction for text-to-video retrieval},
  author={Reddy, Arun and Martin, Alexander and Yang, Eugene and Yates, Andrew and Sanders, Kate and Murray, Kenton and Kriz, Reno and De Melo, Celso M and Van Durme, Benjamin and Chellappa, Rama},
  booktitle={2025 IEEE/CVF Conference on Computer Vision and Pattern Recognition (CVPR)},
  pages={19691--19701},
  year={2025},
  organization={IEEE}
}

@inproceedings{clavie2026lir,
  title={Lir: The first workshop on late interaction and multi vector retrieval@ ecir 2026},
  author={Clavi{\'e}, Benjamin and Li, Xianming and Chaffin, Antoine and Khattab, Omar and Aarsen, Tom and Faysse, Manuel and Li, Jing},
  booktitle={European Conference on Information Retrieval},
  pages={158--168},
  year={2026},
  organization={Springer}
}

@article{clavie2024reducing,
  title={Reducing the footprint of multi-vector retrieval with minimal performance impact via token pooling},
  author={Clavi{\'e}, Benjamin and Chaffin, Antoine and Adams, Griffin},
  journal={arXiv preprint arXiv:2409.14683},
  year={2024}
}

@inproceedings{ma2025lightcolpali,
  title     = {Towards Storage-Efficient Visual Document Retrieval: An Empirical Study on
               Reducing Patch-Level Embeddings},
  author    = {Ma, Yubo and Li, Jinsong and Zang, Yuhang and Wu, Xiaobao and Dong, Xiaoyi and
               Zhang, Pan and Cao, Yuhang and Duan, Haodong and Wang, Jiaqi and Cao, Yixin and
               Sun, Aixin},
  booktitle = {Findings of the Association for Computational Linguistics: ACL 2025},
  pages     = {19568--19580},
  address   = {Vienna, Austria},
  publisher = {Association for Computational Linguistics},
  year      = {2025},
  note      = {arXiv:2506.04997}
}

@misc{liu2026anchor,
  title        = {Structural Anchor Pruning: Training-Free Multi-Vector Compression for
                  Visual Document Retrieval},
  author       = {Liu, Zhuchenyang and Hu, Ziyu and Zhang, Yao and Xiao, Yu},
  howpublished = {arXiv:2601.20107},
  year         = {2026}
}

@article{mace2025vidore,
  title={Vidore benchmark v2: Raising the bar for visual retrieval},
  author={Mac{\'e}, Quentin and Loison, Ant{\'o}nio and Faysse, Manuel},
  journal={arXiv preprint arXiv:2505.17166},
  year={2025}
}

@inproceedings{yan2026sculpting,
  title={Sculpting the Vector Space: Towards Efficient Multi-Vector Visual Document Retrieval via Prune-then-Merge Framework},
  author={Yan, Yibo and Ou, Mingdong and Cao, Yi and Zou, Xin and Huo, Jiahao and Liu, Shuliang and Kwok, James and Hu, Xuming},
  booktitle={Findings of the Association for Computational Linguistics: ACL 2026},
  pages={24883--24925},
  year={2026}
}

@article{veneroso2025crisp,
  title={Crisp: Clustering multi-vector representations for denoising and pruning},
  author={Veneroso, Jo{\~a}o and Jayaram, Rajesh and Rao, Jinmeng and {\'A}brego, Gustavo Hern{\'a}ndez and Hadian, Majid and Cer, Daniel},
  journal={arXiv preprint arXiv:2505.11471},
  year={2025}
}

@inproceedings{macavaney2025efficient,
  title={Efficient constant-space multi-vector retrieval},
  author={MacAvaney, Sean and Mallia, Antonio and Tonellotto, Nicola},
  booktitle={European Conference on Information Retrieval},
  pages={237--245},
  year={2025},
  organization={Springer}
}

@article{dhulipala2024muvera,
  title={Muvera: Multi-vector retrieval via fixed dimensional encoding},
  author={Dhulipala, Laxman and Hadian, Majid and Jayaram, Rajesh and Lee, Jason and Mirrokni, Vahab},
  journal={Advances in Neural Information Processing Systems},
  volume={37},
  pages={101042--101073},
  year={2024}
}

@inproceedings{nardini2024efficient,
  title={Efficient multi-vector dense retrieval with bit vectors},
  author={Nardini, Franco Maria and Rulli, Cosimo and Venturini, Rossano},
  booktitle={European Conference on Information Retrieval},
  pages={3--17},
  year={2024},
  organization={Springer}
}

@article{jayaram2026multi,
  title={Multi-Vector Embeddings are Provably More Expressive than Single Vector Embeddings},
  author={Jayaram, Rajesh},
  journal={arXiv preprint arXiv:2606.23475},
  year={2026}
}

@article{moreira2026nemotron,
  title={Nemotron ColEmbed V2: Top-Performing Late Interaction Embedding Models for Visual Document Retrieval},
  author={Moreira, Gabriel de Souza P and Ak, Ronay and Xu, Mengyao and Holworthy, Oliver and Schifferer, Benedikt and Yu, Zhiding and Babakhin, Yauhen and Osmulski, Radek and Cai, Jiarui and Chesler, Ryan and others},
  journal={arXiv preprint arXiv:2602.03992},
  year={2026}
}

@inproceedings{xiao2026metaembed,
  title={Metaembed: Scaling multimodal retrieval at test-time with flexible late interaction},
  author={Xiao, Zilin and Ma, Qi and Gu, Mengting and Chen, Chun-cheng and Chen, Xintao and Ordonez, Vicente and Mohan, Vijai},
  booktitle={International Conference on Learning Representations},
  volume={2026},
  pages={23619--23638},
  year={2026}
}

@article{xiang2026mm,
  title={MM-Matryoshka: Towards Budget-Elastic Visual Document Retrieval via a 2D Multimodal Matryoshka Training Framework},
  author={Xiang, Haowen and Yan, Yibo and Huo, Jiahao and Huang, Yu and Cao, Yi and Ou, Mingdong and Hu, Xuming},
  journal={arXiv preprint arXiv:2606.07654},
  year={2026}
}

@article{kusupati2022matryoshka,
  title={Matryoshka representation learning},
  author={Kusupati, Aditya and Bhatt, Gantavya and Rege, Aniket and Wallingford, Matthew and Sinha, Aditya and Ramanujan, Vivek and Howard-Snyder, William and Chen, Kaifeng and Kakade, Sham and Jain, Prateek and others},
  journal={Advances in Neural Information Processing Systems},
  volume={35},
  pages={30233--30249},
  year={2022}
}

@article{facco2017estimating,
  title={Estimating the intrinsic dimension of datasets by a minimal neighborhood information},
  author={Facco, Elena and d’Errico, Maria and Rodriguez, Alex and Laio, Alessandro},
  journal={Scientific reports},
  volume={7},
  number={1},
  pages={12140},
  year={2017},
  publisher={Nature Publishing Group UK London}
}
\bibliographystyle{iclr2027_conference}

\newpage
\appendix
\section{Implementation details}\label{app:impl}

\textbf{Architecture.} The refiner is one four-head cross-attention layer with a zero-initialized output projection. The decoder is a two-layer MLP of width 256 with the bounded step of Section~\ref{sec:method} capped at $\alpha = 0.75$. Together they hold 415K parameters. Per-page $k$-means uses two restarts.

\textbf{Optimization.} We train with AdamW at learning rate $2 \times 10^{-4}$ and weight decay $10^{-4}$, cosine decay to 5\% of the initial rate over 100 epochs, and gradient clipping at 1.0. A batch is eight queries, each paired with its relevant page and seven hard negatives drawn from the fifty highest-scoring non-relevant pages under the frozen code. Loss weights are 1.0 for the MaxSim and listwise terms on the regenerated set, 0.5 for the same two terms on the code, 0.5 for the Chamfer and overshoot terms, and 0.1 for the support term, with a listwise temperature of 0.07, 128 fixed support directions, and 48 sampled patches per cluster for the Chamfer term.

\textbf{Model selection and inference.} Two checkpoints are kept from each run, the best single-stage code and the best cascade on that holdout, and each is reported on the metric it is deployed under. The shortlist is $L = 20$ and pages outside it keep their first-stage order. Stored vectors are bfloat16 with one fp16 norm and one uint16 count per cluster. One budget fits in 2.7 minutes on a single NVIDIA A100 80GB, and seeds are 0, 1, and 2.

\section{Evaluation protocol details}\label{app:protocol}

No learned component ever sees a ViDoRe test query or page, so we evaluate on all queries of each subset: 3{,}943 queries in total, of which 1{,}663 are on TAT-DQA. The corpus of each subset is its set of unique pages, de-duplicated by image identity, with query-less pages kept as distractors. Our pipeline reproduces the official leaderboard average for ColPali v1.3. Margins are reported against the strongest training-free baseline chosen per subset and per budget. Learned rows are means over three training seeds with seed standard error at most 0.002 on every cell and margin.

\section{Light-ColPali reproduction}\label{app:lcp}

Light-ColPali~\citep{ma2025lightcolpali} merges post-projector embeddings by clustering and fine-tunes the encoder through the merge. We reproduce the fine-tuning half on the same source as \method. The merge follows the authors: per-page cluster assignment is computed without gradient, and the cluster means are renormalized and trained end to end. Because \texttt{vidore/colpali-v1.3} is itself a PEFT checkpoint, we merge and unload its own adapter before attaching a fresh LoRA (rank 32, $\alpha = 32$, dropout 0.05). We verify at initialization that the untrained merge scores exactly as stock ColPali plus merging before any update. Training uses the pairwise loss of the ColPali training recipe at batch 8 with gradient accumulation 4, five epochs over 4{,}000 query-page pairs and 10\% warmup. One adapter is trained per budget, 13.3M trainable parameters and about 1.5 GPU-hours each on one A100, and evaluated zero-shot on all queries of the ten ViDoRe v1 subsets under the same protocol as \method. The published method trains on 130K queries for about 72 GPU-hours per budget and reports results only at $k \gtrsim 16$, so Table~\ref{tab:lcp} is a statement about what a matched small budget buys and not about the published method.

\section{Per-subset results}\label{app:persubset}

Tables~\ref{tab:app-v1a}, \ref{tab:app-v1b}, and \ref{tab:app-v2} give every cell behind the macro averages of Tables~\ref{tab:main} and \ref{tab:v2}.

\begin{table}[t]
\centering
\small
\caption{ViDoRe v1 per subset, part 1 of 2. nDCG@5 over all queries, three-seed means for the learned rows. Best method per column in bold; the uncompressed row is the ceiling and does not compete.}
\label{tab:app-v1a}
\begin{tabular}{lcccccc}
\toprule
\textbf{Method} & $k{=}2$ & $k{=}4$ & $k{=}8$ & $k{=}16$ & $k{=}32$ & $k{=}64$ \\
\midrule
\multicolumn{7}{l}{\emph{ArxivQA}} \\
\quad Raw $k$-means & 0.451 & 0.561 & 0.669 & 0.710 & 0.760 & 0.794 \\
\quad Token pooling & 0.527 & 0.572 & 0.641 & 0.671 & 0.698 & 0.731 \\
\quad Cluster merging & 0.495 & 0.537 & 0.646 & 0.723 & 0.779 & 0.800 \\
\quad \method, training-free stage & 0.584 & 0.666 & 0.722 & 0.742 & 0.786 & 0.806 \\
\quad \method, stage 1 & 0.610 & 0.668 & 0.729 & 0.760 & 0.785 & 0.810 \\
\quad \method, full & \textbf{0.617} & \textbf{0.710} & \textbf{0.747} & \textbf{0.769} & \textbf{0.803} & \textbf{0.812} \\
\quad Uncompressed & 0.827 & 0.827 & 0.827 & 0.827 & 0.827 & 0.827 \\
\midrule
\multicolumn{7}{l}{\emph{DocVQA}} \\
\quad Raw $k$-means & 0.118 & 0.183 & 0.255 & 0.335 & 0.400 & 0.460 \\
\quad Token pooling & 0.191 & 0.220 & 0.250 & 0.307 & 0.343 & 0.387 \\
\quad Cluster merging & 0.140 & 0.183 & 0.269 & 0.343 & 0.411 & 0.470 \\
\quad \method, training-free stage & 0.204 & 0.271 & 0.344 & 0.413 & 0.466 & \textbf{0.500} \\
\quad \method, stage 1 & 0.238 & 0.280 & 0.346 & 0.401 & 0.450 & 0.489 \\
\quad \method, full & \textbf{0.249} & \textbf{0.308} & \textbf{0.370} & \textbf{0.424} & \textbf{0.469} & 0.495 \\
\quad Uncompressed & 0.525 & 0.525 & 0.525 & 0.525 & 0.525 & 0.525 \\
\midrule
\multicolumn{7}{l}{\emph{InfoVQA}} \\
\quad Raw $k$-means & 0.613 & 0.642 & 0.669 & 0.721 & 0.786 & 0.815 \\
\quad Token pooling & 0.684 & 0.705 & 0.721 & 0.747 & 0.765 & 0.791 \\
\quad Cluster merging & 0.620 & 0.603 & 0.641 & 0.686 & 0.744 & 0.792 \\
\quad \method, training-free stage & 0.678 & 0.717 & 0.747 & 0.786 & \textbf{0.818} & \textbf{0.832} \\
\quad \method, stage 1 & 0.696 & 0.721 & 0.753 & 0.788 & 0.812 & 0.830 \\
\quad \method, full & \textbf{0.701} & \textbf{0.736} & \textbf{0.761} & \textbf{0.796} & 0.811 & 0.831 \\
\quad Uncompressed & 0.842 & 0.842 & 0.842 & 0.842 & 0.842 & 0.842 \\
\midrule
\multicolumn{7}{l}{\emph{Shift Project}} \\
\quad Raw $k$-means & 0.297 & 0.335 & 0.412 & 0.473 & 0.544 & 0.602 \\
\quad Token pooling & 0.392 & 0.403 & 0.429 & 0.485 & 0.540 & 0.609 \\
\quad Cluster merging & 0.273 & 0.256 & 0.273 & 0.358 & 0.387 & 0.566 \\
\quad \method, training-free stage & 0.386 & 0.407 & 0.547 & 0.584 & 0.646 & \textbf{0.693} \\
\quad \method, stage 1 & \textbf{0.398} & \textbf{0.451} & \textbf{0.561} & \textbf{0.615} & 0.657 & 0.687 \\
\quad \method, full & 0.396 & 0.451 & 0.549 & 0.610 & \textbf{0.661} & \textbf{0.693} \\
\quad Uncompressed & 0.768 & 0.768 & 0.768 & 0.768 & 0.768 & 0.768 \\
\midrule
\multicolumn{7}{l}{\emph{Synthetic AI}} \\
\quad Raw $k$-means & 0.592 & 0.641 & 0.716 & 0.807 & 0.916 & 0.931 \\
\quad Token pooling & 0.742 & 0.749 & 0.792 & 0.824 & 0.882 & 0.924 \\
\quad Cluster merging & 0.548 & 0.561 & 0.579 & 0.717 & 0.836 & 0.926 \\
\quad \method, training-free stage & 0.706 & 0.756 & 0.810 & 0.908 & 0.938 & 0.960 \\
\quad \method, stage 1 & \textbf{0.779} & \textbf{0.819} & 0.867 & \textbf{0.926} & 0.940 & 0.960 \\
\quad \method, full & 0.757 & 0.810 & \textbf{0.880} & 0.922 & \textbf{0.943} & \textbf{0.966} \\
\quad Uncompressed & 0.984 & 0.984 & 0.984 & 0.984 & 0.984 & 0.984 \\
\bottomrule
\end{tabular}
\end{table}
\begin{table}[t]
\centering
\small
\caption{ViDoRe v1 per subset, part 2 of 2. nDCG@5 over all queries, three-seed means for the learned rows. Best method per column in bold; the uncompressed row is the ceiling and does not compete.}
\label{tab:app-v1b}
\begin{tabular}{lcccccc}
\toprule
\textbf{Method} & $k{=}2$ & $k{=}4$ & $k{=}8$ & $k{=}16$ & $k{=}32$ & $k{=}64$ \\
\midrule
\multicolumn{7}{l}{\emph{Synthetic Energy}} \\
\quad Raw $k$-means & 0.551 & 0.613 & 0.666 & 0.791 & 0.890 & 0.920 \\
\quad Token pooling & 0.691 & 0.710 & 0.732 & 0.735 & 0.812 & 0.873 \\
\quad Cluster merging & 0.562 & 0.571 & 0.653 & 0.720 & 0.811 & 0.897 \\
\quad \method, training-free stage & 0.662 & 0.704 & 0.814 & 0.871 & \textbf{0.924} & \textbf{0.929} \\
\quad \method, stage 1 & 0.736 & 0.747 & 0.830 & \textbf{0.892} & 0.923 & \textbf{0.929} \\
\quad \method, full & \textbf{0.739} & \textbf{0.786} & \textbf{0.843} & 0.890 & 0.920 & 0.926 \\
\quad Uncompressed & 0.938 & 0.938 & 0.938 & 0.938 & 0.938 & 0.938 \\
\midrule
\multicolumn{7}{l}{\emph{Synthetic Gov.}} \\
\quad Raw $k$-means & 0.514 & 0.585 & 0.691 & 0.741 & 0.852 & 0.905 \\
\quad Token pooling & 0.639 & 0.713 & 0.747 & 0.791 & 0.873 & 0.899 \\
\quad Cluster merging & 0.465 & 0.526 & 0.568 & 0.671 & 0.777 & 0.833 \\
\quad \method, training-free stage & 0.624 & 0.705 & 0.782 & 0.825 & 0.912 & \textbf{0.940} \\
\quad \method, stage 1 & \textbf{0.732} & 0.765 & 0.804 & 0.858 & 0.918 & 0.935 \\
\quad \method, full & 0.712 & \textbf{0.783} & \textbf{0.835} & \textbf{0.878} & \textbf{0.923} & 0.931 \\
\quad Uncompressed & 0.950 & 0.950 & 0.950 & 0.950 & 0.950 & 0.950 \\
\midrule
\multicolumn{7}{l}{\emph{Synthetic Health}} \\
\quad Raw $k$-means & 0.683 & 0.701 & 0.804 & 0.856 & 0.890 & 0.938 \\
\quad Token pooling & 0.771 & 0.802 & 0.860 & 0.869 & 0.904 & 0.922 \\
\quad Cluster merging & 0.621 & 0.650 & 0.719 & 0.759 & 0.835 & 0.908 \\
\quad \method, training-free stage & 0.757 & 0.780 & 0.868 & 0.901 & 0.932 & 0.955 \\
\quad \method, stage 1 & 0.821 & 0.855 & 0.910 & 0.917 & 0.940 & \textbf{0.963} \\
\quad \method, full & \textbf{0.822} & \textbf{0.869} & \textbf{0.922} & \textbf{0.922} & \textbf{0.945} & 0.962 \\
\quad Uncompressed & 0.976 & 0.976 & 0.976 & 0.976 & 0.976 & 0.976 \\
\midrule
\multicolumn{7}{l}{\emph{TabFQuAD}} \\
\quad Raw $k$-means & 0.601 & 0.617 & 0.723 & 0.792 & 0.832 & 0.844 \\
\quad Token pooling & 0.606 & 0.652 & 0.700 & 0.734 & 0.760 & 0.801 \\
\quad Cluster merging & 0.582 & 0.650 & 0.739 & 0.811 & 0.833 & 0.850 \\
\quad \method, training-free stage & 0.621 & 0.691 & 0.768 & 0.822 & \textbf{0.841} & 0.852 \\
\quad \method, stage 1 & 0.628 & 0.710 & 0.784 & 0.827 & \textbf{0.841} & 0.853 \\
\quad \method, full & \textbf{0.643} & \textbf{0.716} & \textbf{0.795} & \textbf{0.832} & 0.837 & \textbf{0.854} \\
\quad Uncompressed & 0.852 & 0.852 & 0.852 & 0.852 & 0.852 & 0.852 \\
\midrule
\multicolumn{7}{l}{\emph{TAT-DQA}} \\
\quad Raw $k$-means & 0.222 & 0.240 & 0.333 & 0.397 & 0.501 & 0.578 \\
\quad Token pooling & 0.283 & 0.319 & 0.369 & 0.403 & 0.448 & 0.536 \\
\quad Cluster merging & 0.166 & 0.173 & 0.226 & 0.345 & 0.483 & 0.585 \\
\quad \method, training-free stage & 0.297 & 0.356 & 0.434 & 0.507 & 0.577 & 0.626 \\
\quad \method, stage 1 & 0.326 & 0.395 & 0.468 & 0.531 & 0.587 & 0.633 \\
\quad \method, full & \textbf{0.338} & \textbf{0.401} & \textbf{0.480} & \textbf{0.544} & \textbf{0.599} & \textbf{0.639} \\
\quad Uncompressed & 0.700 & 0.700 & 0.700 & 0.700 & 0.700 & 0.700 \\
\bottomrule
\end{tabular}
\end{table}
\begin{table}[t]
\centering
\small
\caption{ViDoRe v2 per subset. Graded nDCG@5 over all queries. Best method per column in bold; the uncompressed row is the ceiling and does not compete.}
\label{tab:app-v2}
\begin{tabular}{lcccccc}
\toprule
\textbf{Method} & $k{=}2$ & $k{=}4$ & $k{=}8$ & $k{=}16$ & $k{=}32$ & $k{=}64$ \\
\midrule
\multicolumn{7}{l}{\emph{Biomedical lectures}} \\
\quad Raw $k$-means & 0.187 & 0.243 & 0.321 & 0.378 & 0.446 & 0.508 \\
\quad Token pooling & 0.296 & 0.322 & 0.347 & 0.388 & 0.422 & 0.460 \\
\quad Cluster merging & 0.251 & 0.280 & 0.310 & 0.361 & 0.459 & 0.512 \\
\quad \method, training-free stage & 0.285 & 0.356 & 0.419 & 0.461 & 0.510 & 0.548 \\
\quad \method, stage 1 & 0.327 & 0.390 & 0.440 & 0.472 & 0.521 & 0.550 \\
\quad \method, full & \textbf{0.327} & \textbf{0.408} & \textbf{0.458} & \textbf{0.496} & \textbf{0.527} & \textbf{0.557} \\
\quad Uncompressed & 0.564 & 0.564 & 0.564 & 0.564 & 0.564 & 0.564 \\
\midrule
\multicolumn{7}{l}{\emph{Economics reports}} \\
\quad Raw $k$-means & 0.350 & 0.323 & 0.331 & 0.361 & 0.378 & 0.372 \\
\quad Token pooling & 0.369 & 0.343 & 0.323 & 0.345 & 0.357 & 0.392 \\
\quad Cluster merging & 0.335 & 0.323 & 0.278 & 0.233 & 0.255 & 0.285 \\
\quad \method, training-free stage & 0.308 & 0.344 & 0.377 & 0.405 & 0.403 & 0.415 \\
\quad \method, stage 1 & 0.354 & 0.386 & 0.404 & 0.424 & 0.412 & \textbf{0.433} \\
\quad \method, full & \textbf{0.396} & \textbf{0.407} & \textbf{0.405} & \textbf{0.447} & \textbf{0.427} & 0.431 \\
\quad Uncompressed & 0.468 & 0.468 & 0.468 & 0.468 & 0.468 & 0.468 \\
\midrule
\multicolumn{7}{l}{\emph{ESG reports (human)}} \\
\quad Raw $k$-means & 0.082 & 0.082 & 0.219 & 0.263 & 0.347 & 0.375 \\
\quad Token pooling & 0.132 & 0.181 & 0.179 & 0.208 & 0.235 & 0.347 \\
\quad Cluster merging & 0.098 & 0.057 & 0.068 & 0.143 & 0.250 & 0.364 \\
\quad \method, training-free stage & 0.157 & 0.188 & 0.332 & 0.364 & 0.430 & 0.467 \\
\quad \method, stage 1 & \textbf{0.184} & 0.215 & \textbf{0.377} & 0.404 & 0.434 & 0.455 \\
\quad \method, full & 0.174 & \textbf{0.228} & 0.349 & \textbf{0.409} & \textbf{0.444} & \textbf{0.485} \\
\quad Uncompressed & 0.523 & 0.523 & 0.523 & 0.523 & 0.523 & 0.523 \\
\midrule
\multicolumn{7}{l}{\emph{ESG reports}} \\
\quad Raw $k$-means & 0.119 & 0.156 & 0.227 & 0.272 & 0.334 & 0.364 \\
\quad Token pooling & 0.114 & 0.116 & 0.128 & 0.164 & 0.214 & 0.282 \\
\quad Cluster merging & 0.069 & 0.065 & 0.102 & 0.142 & 0.178 & 0.293 \\
\quad \method, training-free stage & 0.132 & 0.214 & 0.304 & 0.359 & 0.442 & 0.467 \\
\quad \method, stage 1 & 0.161 & 0.248 & 0.318 & 0.370 & 0.441 & 0.473 \\
\quad \method, full & \textbf{0.172} & \textbf{0.276} & \textbf{0.340} & \textbf{0.406} & \textbf{0.488} & \textbf{0.487} \\
\quad Uncompressed & 0.514 & 0.514 & 0.514 & 0.514 & 0.514 & 0.514 \\
\bottomrule
\end{tabular}
\end{table}

\end{document}